\documentclass[11pt,english]{article}
\usepackage[T1]{fontenc}
\usepackage[latin9]{inputenc}
\usepackage{verbatim}

\newcommand{\noun}[1]{\textsc{#1}}
\newcommand{\binom}[2]{{#1 \choose #2}}

\usepackage{amsmath}
\usepackage{amssymb}
\usepackage{amsthm}
\usepackage{amsfonts}
\usepackage{fullpage}
\usepackage{color}
\usepackage{times}
\usepackage{babel}

\newtheorem{theorem}{Theorem}
\newtheorem{lem}[theorem]{Lemma}
\newtheorem{fact}[theorem]{Fact}

\newtheorem{claim}[theorem]{Claim}

\newtheorem{prop}[theorem]{Proposition}

\theoremstyle{definition}
\newtheorem{defn}{Definition}

\theoremstyle{remark}
\newtheorem*{rem}{Remark}

\newcommand{\eps}{\varepsilon}
\newcommand{\D}{{\cal D}}
\newcommand{\E}{{\cal E}}

\newcommand{\cube}{\operatorname{\{0, 1\}}}
\newcommand{\Z}{\mathbb{Z}}
\newcommand{\polyn}{\operatorname{poly}}
\newcommand{\polylog}{\operatorname{polylog} }
\newcommand{\mem}{\mbox{\textsc{Mem}}}
\newcommand{\polyeval}{\mbox{\textsc{PolyEval}}}

\date{}

\begin{document}

\title{Efficient and Error-Correcting Data Structures for Membership and
Polynomial Evaluation}

\author{Victor Chen%
\thanks{MIT CSAIL, victor@csail.mit.edu. Supported by NSF award CCF-0829672. %
} \and Elena Grigorescu%
\thanks{MIT CSAIL, elena\_g@mit.edu. This work started when this author was
visiting CWI in Summer 2008. Supported by NSF award CCF-0829672. %
} \and Ronald de Wolf%
\thanks{CWI Amsterdam, rdewolf@cwi.nl. Supported by a Vidi grant
from the Netherlands Organization for Scientific Research (NWO).
}}

\maketitle
\setcounter{page}{0} 
\begin{abstract}
We construct efficient data structures that are resilient against
a constant fraction of adversarial noise. Our model requires that
the decoder answers \emph{most} queries correctly with high probability
and for the remaining queries, the decoder with high probability either
answers correctly or declares ``don't know.'' Furthermore, if
there is no noise on the data structure, it answers \emph{all} queries
correctly with high probability. Our model is the common generalization
of an error-correcting data structure model proposed recently by de~Wolf, 
and the notion of ``relaxed locally decodable codes'' developed in the PCP literature.

We measure the efficiency of a data structure in terms of its \emph{length},
(the number of bits in its representation), and query-answering
time, measured by the number of \emph{bit-probes} to the (possibly corrupted)
representation. We obtain results for the following two data structure problems: 
\begin{itemize}
\item (Membership) Store a subset $S$ of size at most $s$ from a universe
of size $n$ such that membership queries can be answered efficiently,
i.e., decide if a given element from the universe is in $S$. \\
 We construct an error-correcting data structure for this problem
with length nearly linear in $s\log n$ that answers membership queries
with $O(1)$ bit-probes. This nearly matches the asymptotically optimal
parameters for the noiseless case: length $O(s\log n)$ and one bit-probe,
due to Buhrman, Miltersen, Radhakrishnan, and Venkatesh. 
\item (Univariate polynomial evaluation) Store a univariate polynomial $g$
of degree $\deg(g)\leq s$ over the integers modulo $n$ such that
evaluation queries can be answered efficiently, i.e., evaluate the
output of $g$ on a given integer modulo $n$. \\
 We construct an error-correcting data structure for this problem
with length nearly linear in $s\log n$ that answers evaluation queries
with $\polylog s\cdot\log^{1+o(1)}n$ bit-probes. This nearly matches
the parameters of the best-known noiseless construction, due to Kedlaya
and Umans.
\end{itemize}

\end{abstract}
\thispagestyle{empty} \newpage{}

\section{Introduction}
The area of data structures is one of the oldest and most fundamental
parts of computer science, in theory as well as in practice. The
underlying question is a time-space tradeoff: we are given a piece
of data, and we would like to store it in a short, space-efficient
data structure that allows us to quickly answer specific queries about
the stored data. On one extreme, we can store the data as just
a list of the correct answers to all possible queries. This is extremely
time-efficient (one can immediately look up the correct answer without
doing any computation) but usually takes significantly more space
than the information-theoretic minimum. At the other extreme, we can
store a maximally compressed version of the data. This method is extremely
space-efficient but not very time-efficient since one usually has to undo
the whole compression first. A good data structure sits somewhere
in the middle: it does not use much more space than the information-theoretic
minimum, but it also stores the data in a structured way that enables
efficient query-answering.

It is reasonable to assume that most practical implementations of
data storage are susceptible to \emph{noise}: over time some of the
information in the data structure may be corrupted or erased by various
accidental or malicious causes. This buildup of errors may cause the
data structure to deteriorate so that most queries are not answered
correctly anymore. Accordingly, it is a natural task to design data
structures that are not only efficient in space and time but also
resilient against a certain amount of \emph{adversarial} noise, where
the noise can be placed in positions that make decoding as difficult
as possible. 

Ways to protect information and computation against noise
have been well studied in the theory of error-correcting codes
and of fault-tolerant computation. In the data structure
literature, constructions under often incomparable models have been
designed to cope with noise, and we examine a few of these models.
Aumann and Bender~\cite{aumann&bender:ftdata} studied pointer-based
data structures such as linked lists, stacks, and binary search trees.
In this model, errors (adversarial but detectable) occur whenever
all the pointers from a node are lost. They measure the dependency
between the number of errors and the number of nodes that become irretrievable,
and designed a number of efficient data structures where this dependency is reasonable.

Another model for studying data structures with noise is the faulty-memory RAM model, 
introduced by Finocchi and Italiano~\cite{finocchi&italiano:faults}. 
In a faulty-memory RAM, there are $O(1)$ memory cells that cannot
be corrupted by noise. Elsewhere, errors (adversarial and undetectable)
may occur at any time, even during the decoding procedure.
Many data structure problems have been examined in this model, 
such as sorting~\cite{fgi:resilientsorting}, searching~\cite{fgi:resilientsearch},
priority queues~\cite{jmm:resilientpriority} and dictionaries~\cite{bffgijmm:resilientdict}.
However, the number of errors that can be tolerated is typically less than
a linear portion of the size of the input. Furthermore, correctness
can only be guaranteed for keys that are not affected by noise. For
instance, for the problem of comparison-sorting on $n$ keys, the
authors of~\cite{fgi:resilientsorting} designed a resilient sorting
algorithm that tolerates $\sqrt{n\log n}$ keys being corrupted and
ensures that the set of uncorrupted keys remains sorted.

Recently, de~Wolf~\cite{wolf:ecdata} considered another model of
resilient data structures. The representation of the data structure
is viewed as a bit-string, from which a decoding procedure can read
any particular set of bits to answer a data query. The representation
must be able to tolerate a constant fraction $\delta$ of adversarial noise in the bit-string%
\footnote{We only consider bit-flip-errors here, not erasures. Since erasures
are easier to deal with than bit-flips, it suffices to design a data
structure dealing with bit-flip-errors.} 
(but not inside the decoding procedure). His model generalizes the
usual noise-free data structures (where $\delta=0$) as well as the
so-called ``locally decodable codes'' (LDCs)~\cite{katz&trevisan:ldc}.
Informally, an LDC is an encoding that is tolerant of noise and allows
fast decoding so that each message symbol can be retrieved correctly
with high probability. Using LDCs as building blocks, de~Wolf constructed
data structures for several problems. 

Unfortunately, de~Wolf's model has the drawback that the optimal
time-space tradeoffs are much worse than in the noise-free model.
The reason is that all known constructions of LDCs that make $O(1)$
bit-probes~\cite{yekhanin:3ldcj,efremenko:ldc} have very poor encoding
length (super-polynomial in the message length). In fact, the encoding
length provably must be super-linear in the message length~\cite{katz&trevisan:ldc,kerenidis&wolf:qldcj,woodruff:ldclower}.
As his model is a generalization of LDCs, data structures cannot have
a succinct representation that has length proportional to the information-theoretic
bound.

We thus ask: what is a clean model of data structures that allows
efficient representations \emph{and} has error-correcting capabilities?
Compared with the pointer-based model and the faulty-memory RAM, de~Wolf's
model imposes a rather stringent requirement on decoding: \emph{every}
query must be answered correctly with high probability from the possibly corrupted
encoding. While this requirement is crucial in the definition of LDCs
due to their connection to complexity theory and cryptography, for
data structures it seems somewhat restrictive.

In this paper, we consider a broader, more relaxed notion of error-correcting for data structures. 
In our model, for most queries, the decoder has to return the correct answer with
high probability.  However, for the few remaining queries, the decoder may claim ignorance, 
i.e., declare the data item unrecoverable from the (corrupted) data structure. 
Still, for \emph{every} query, the answer is incorrect only with small probability.
In fact, just as de~Wolf's model is a generalization of LDCs, our model in this paper
is a generalization of the ``relaxed'' locally decodable codes (RLDCs) introduced
by Ben-Sasson, Goldreich, Harsha, Sudan, and Vadhan~\cite{bghsv04}.
They relax the usual definition of an LDC by requiring the decoder
to return the correct answer on \emph{most} rather than all queries.
For the remaining queries it is allowed to claim ignorance, i.e.,
to output a special symbol `$\perp$' interpreted as ``don't know''
or ``unrecoverable.'' As shown in~\cite{bghsv04}, relaxing the
LDC-definition like this allows for constructions of RLDCs
with $O(1)$ bit-probes of \emph{nearly linear} length. 

Using RLDCs as building blocks, we construct error-correcting data structures that are very efficient in terms of time
as well as space. Before we describe our results, let us define
our model formally. First, a \emph{data structure problem} is specified
by a set $D$ of \emph{data items}, a set $Q$ of \emph{queries},
a set $A$ of \emph{answers}, and a function $f:D\times Q\rightarrow A$
which specifies the correct answer $f(x,q)$ of query $q$ to data
item~$x$. A data structure for $f$ is specified by four parameters:
$t$ the number bit-probes, $\delta$ the fraction of noise, $\eps$ an upper
bound on the error probability for each query, and $\lambda$ an upper bound on
the fraction of queries in $Q$ that are not answered correctly with high probability 
(the `$\lambda$' stands for ``lost'').

\begin{defn}\label{def:data structure}Let $f:D\times Q\rightarrow A$
be a data structure problem. Let $t>0$ be an integer, 
$\delta\in[0,1]$, $\eps\in[0,1/2]$, and $\lambda\in[0,1]$. We
say that $f$ has a \emph{$(t,\delta,\eps,\lambda)$-data structure}
of length $N$ if there exist an encoder $\E:D\rightarrow\cube^{N}$ and a
(randomized) decoder $\D$ with the following properties: for every $x\in D$ and every $w\in\cube^{N}$ at Hamming distance $\Delta(w,\E(x))\leq\delta N$, 
\begin{enumerate}
\item $\D$ makes at most $t$ bit-probes to $w$, 
\item $\Pr[\D^{w}(q)\in\{f(x,q),\perp\}]\geq1-\eps$ for every $q\in Q$, 
\item the set $G=\{q:\Pr[\D^{w}(q)=f(x,q)]\geq1-\eps\}$ has size at least
$(1-\lambda)|Q|$ (`$G$' stands for ``good''),
\item if $w=\E(x)$, then $G=Q$. 
\end{enumerate}
\end{defn}

Here $\D^{w}(q)$ denotes the random variable which is the decoder's output on inputs $w$ and $q$.
The notation indicates that it accesses the two inputs in different ways: 
while it has full access to the query $q$, it only has bit-probe access (or ``oracle access'') to the string $w$.

We say that a $(t,\delta,\eps,\lambda)$-data structure is \emph{error-correcting},
or an \emph{error-correcting data structure}, if $\delta>0$. 
Setting $\lambda=0$ recovers the original notion of error-correction
in de~Wolf's model~\cite{wolf:ecdata}. A \emph{$(t,\delta,\eps,\lambda)$-relaxed locally
decodable code (RLDC)}, defined in~\cite{bghsv04}, is an error-correcting
data structure for the membership function $f:\cube^{n}\times[n]\rightarrow\cube$,
where $f(x,i)=x_{i}$. A \emph{$(t,\delta,\eps)$-locally decodable code (LDC)}, defined
by Katz and Trevisan~\cite{katz&trevisan:ldc}, is an RLDC with $\lambda=0$.

\begin{rem}For the data structure problems considered in this paper,
our decoding procedures make only \emph{non-adaptive} probes, i.e.,
the positions of the probes are determined all at once and sent simultaneously
to the oracle. For other data structure problems it may be natural
for decoding procedures to be adaptive. Thus, we do not require $\D$
to be non-adaptive in Condition~1 of Definition~\ref{def:data structure}.\end{rem}

\subsection{Our results\label{sub:Our results}}

We obtain efficient error-correcting data structures for the following two data structure problems.

\paragraph{\textbf{\noun{Membership:}}}

Consider a universe $[n]=\{1,\ldots,n\}$ and some nonnegative integer
$s\leq n$. Given a set $S\subseteq[n]$ with at most $s$ elements,
one would like to store $S$ in a compact representation that can
answer ``membership queries'' efficiently, i.e., given an index
$i\in[n]$, determine whether or not $i\in S$. Formally $D=\{S:S\subseteq[n],|S|\leq s\}$,
$Q=[n]$, and $A=\{0,1\}$. The function $\mem_{n,s}(S,i)$ is $1$
if $i\in S$ and $0$ otherwise.

Since there are at least $\binom{n}{s}$ subsets of the universe of
size at most $s$, each subset requiring a different instantiation
of the data structure, the information-theoretic lower bound on the
space of any data structure is at least $\log\binom{n}{s}\approx s\log n$
bits.%
\footnote{Our logs are always to base~$2$.}
An easy way to achieve this is to store $S$ in sorted order. If
each number is stored in its own $\log n$-bit ``cell,'' this
data structure takes $s$ cells, which is $s\log n$ bits. To answer
a membership query, one can do a binary search on the list to determine
whether $i\in S$ using about $\log s$ ``cell-probes,'' or $\log s\cdot\log n$
bit-probes. The length of this data structure is essentially optimal,
but its number of probes is not. Fredman, Koml\'{o}s, and Szemer\'{e}di~\cite{fks:sparsetable}
developed a famous hashing-based data structure that has length $O(s)$
cells (which is $O(s\log n)$ bits) and only needs a \emph{constant}
number of cell-probes (which is $O(\log n)$ bit-probes). Buhrman,
Miltersen, Radhakrishnan, and Venkatesh~\cite{bmrv:bitvectorsj}
improved upon this by designing a data structure of length $O(s\log n)$
bits that answers queries with \emph{only one bit-probe} and a small error probability. 
This is simultaneously optimal in terms of time (clearly one bit-probe cannot
be improved upon) and space (up to a constant factor).

None of the aforementioned data structures can tolerate a constant fraction of noise.
To protect against noise for this problem, de~Wolf~\cite{wolf:ecdata}
constructed an error-correcting data structure with $\lambda=0$ using
a locally decodable code (LDC). That construction
answers membership queries in $t$ bit-probes and has length roughly
$L(s,t)\log n$, where $L(s,t)$ is the shortest length of an LDC encoding
$s$ bits with bit-probe complexity $t$. Currently, all known LDCs with $t=O(1)$ have $L(s,t)$ super-polynomial in $s$~\cite{bikr:improvedpir,yekhanin:3ldcj,efremenko:ldc}.
In fact, $L(s,t)$ must be super-linear for all constant $t$, see e.g.~\cite{katz&trevisan:ldc,kerenidis&wolf:qldcj,woodruff:ldclower}.

Under our present model of error-correction, we can construct much
more efficient data structures with error-correcting capability. First,
it is not hard to show that by composing the BMRV data structure~\cite{bmrv:bitvectorsj}
with the error-correcting data structure for $\mem_{n,n}$ (equivalently,
an RLDC)~\cite{bghsv04}, one can already obtain an error-correcting
data structure of length $O((s\log n)^{1+\eta})$, where $\eta$ is
an arbitrarily small constant. However, following an approach taken
in~\cite{wolf:ecdata}, we obtain a data structure of length $O(s^{1+\eta}\log n)$,
which is much shorter than the aforementioned construction if $s=o(\log n)$.

\begin{theorem}\label{thm:membership} For every $\eps,\eta\in(0,1)$,
there exist an integer $t>0$ and real $\tau>0$, such that for
all $s$ and $n$, and every $\delta\leq\tau$, $\mem_{n,s}$ has
a $(t,\delta,\eps,\frac{s}{2n})$-data structure of length $O(s^{1+\eta}\log n)$.
\end{theorem}

We will prove Theorem~\ref{thm:membership} in Section~\ref{sec:Membership problem}.
Note that the size of the good set $G$ is at least $n-\frac{s}{2}$.
Hence corrupting a $\delta$-fraction of the bits of the data structure
may cause a decoding failure for at most half of the queries $i\in S$
but not all. One may replace this factor $\frac{1}{2}$ easily by
another constant (though the parameters $t$ and $\tau$ will then
change).

\paragraph{\bf{\noun{Polynomial evaluation:}}}

Let $\Z_{n}$ denote the set of integers modulo $n$ and $s\leq n$
be some nonnegative integer. Given a univariate polynomial $g\in\Z_{n}[X]$
of degree at most $s$, we would like to store $g$ in a compact representation
so that for each evaluation query $a\in\Z_{n}$, $g(a)$ can be computed
efficiently. Formally, $D=\{g:g\in\Z_{n}[X],\deg(g)\leq s\}$, $Q=\Z_{n}$,
and $A=\Z_{n}$, and the function is $\polyeval_{n,s}(g,a)=g(a)$.

Since there are $n^{s+1}$ polynomials of degree at most $s$, with
each polynomial requiring a different instantiation of the data structure,
the information-theoretic lower bound on the space of any data structure
for this problem is at least $\log(n^{s+1})\approx s\log n$
bits. Since each answer is an element of $\Z_{n}$ and must be
represented by $\left\lfloor \log n\right\rfloor +1$ bits, $\left\lfloor \log n\right\rfloor +1$
is the information-theoretic lower bound on the bit-probe complexity.

Consider the following two naive solutions. On one hand, one can simply
record the evaluations of $g$ in a table with $n$ entries, each
with $\left\lfloor \log n\right\rfloor +1$ bits. The length of this
data structure is $O(n\log n)$ and each query requires reading only
$\left\lfloor \log n\right\rfloor +1$ bits. On the other hand, 
$g$ can be stored as a table of its $s+1$ coefficients. This gives a data structure of
length and bit-probe complexity $(s+1)(\left\lfloor \log n\right\rfloor +1)$. 

A natural question is whether one can construct a data structure that
is optimal both in terms of space and time, i.e., has length $O(s\log n)$
and answers queries with $O(\log n)$ bit-probes. No such constructions
are known to exist. However, some lower bounds are known in the weaker
cell-probe model, where each cell is a sequence of $\left\lfloor \log n\right\rfloor +1$
bits. For instance, as noted in~\cite{miltersen:cellprobesurvey},
any data structure for  \noun{polynomial evaluation} that stores $O(s^{2})$
cells ($O(s^{2}\log n)$ bits) requires reading at least $\Omega(s)$
cells ($\Omega(s\log n)$ bits). %
 Moreover, by~\cite{miltersen95polyeval}, if $\log n\gg s\log s$
and the data structure is constrained to store $s^{O(1)}$ cells,
then its query complexity is $\Omega(s)$ cells. This implies that
the second trivial construction described above is essentially optimal in the cell-probe model.

Recently, Kedlaya and Umans~\cite{kedl-umans} obtained a data structure
of length $s^{1+\eta}\log^{1+o(1)}n$ (where $\eta$ is an arbitrarily
small constant) and answers evaluation queries with $O(\polylog s\cdot\log^{1+o(1)}n)$
bit-probes. These parameters exhibit the best tradeoff
between $s$ and $n$ so far. When $s=n^{\eta}$ for some $0<\eta<1$,
the data structure of Kedlaya and Umans~\cite{kedl-umans} is much superior
to the trivial solution: its length is nearly optimal, and the query complexity
drops from $\polyn n$ to only $\polylog n$ bit-probes.

Here we construct an error-correcting data structure for the polynomial
evaluation problem that works even in the presence of adversarial noise,
with length nearly linear in $s\log n$ and bit-probe complexity $O(\polylog s\cdot\log^{1+o(1)}n)$. Formally:

\begin{theorem}\label{thm:poly eval} For every $\eps,\lambda,\eta\in(0,1)$,
there exists $\tau\in(0,1)$ such that for all positive integers $s\leq n$,
for all $\delta\leq\tau$, the data structure problem $\polyeval_{n,s}$
has a $(O(\polylog s\cdot\log^{1+o(1)}n),\delta,\eps,\lambda)$-data structure 
of length $O((s\log n)^{1+\eta})$. 
\end{theorem}

\begin{rem}We note that Theorem~\ref{thm:poly eval} easily holds when
$s=(\log n)^{o(1)}$. As we discussed previously, one can just store a
table of the $s+1$ coefficients of $g$. To make this error-correcting, 
encode the entire table by a standard error-correcting code. 
This has length and bit-probe complexity $O(s\log n)=O(\log^{1+o(1)}n$).
\end{rem}

\subsection{Our techniques\label{sub:Our-technique}}

At a high level, for both data structure problems we build
our constructions by composing a relaxed locally decodable code with
an appropriate noiseless data structure. If the underlying probe-accessing
scheme in a noiseless data structure is ``pseudorandom,'' then
the noiseless data structure can be made error-correcting by appropriate
compositions with other data structures. By pseudorandom, we mean
that if a query is chosen uniformly at random from $Q$, then the positions of the probes selected
also ``behave'' as if they are chosen uniformly at random. Such
property allows us to analyze the error-tolerance of our constructions.

More specifically, for the \noun{membership} problem we build upon the noiseless data structure of Buhrman et al.~\cite{bmrv:bitvectorsj}.  While de Wolf~\cite{wolf:ecdata} combined this with LDCs to get a rather long data structure with $\lambda=0$, we will combine it here with RLDCs to get nearly optimal length with small (but non-zero) $\lambda$.  In order to bound $\lambda$ in our new construction, we make use of the fact that the \cite{bmrv:bitvectorsj}-construction is a bipartite \emph{expander graph}, as explained below after Theorem~\ref{thm:bmrv}.  This property wasn't needed in~\cite{wolf:ecdata}. The left side of the expander represents the set of queries, and a neighborhood of a query (a left node) represents the set of possible bit-probes that can be chosen to answer this query. The expansion property of the graph essentially implies that for a random query, the distribution of a bit-probe chosen to answer this query is close to uniform.\footnote{We remark that this is different from the notion of smooth decoding in the LDC literature, which requires that for every {\em fixed} query, each bit-probe by itself is chosen with probability close to uniform (though not independent of the other bit-probes).} 
This property allows us to construct an efficient, error-correcting data structure for this problem.

For the polynomial evaluation problem, we rely upon the noiseless data structure of Kedlaya and Umans~\cite{kedl-umans}, which has a decoding procedure that uses the reconstructive algorithm from the Chinese Remainder Theorem.  
The property that we need is the simple fact that
if $a$ is chosen uniformly at random from $\Z_{n}$, then for any
$m\leq n$, $a$ modulo $m$ is uniformly distributed in $\Z_{m}$.
This implies that for a random evaluation point $a$, the distribution of certain tuples of cell-probes used to answer this evaluation point is close to uniform. This observation allows us to construct an efficient, error-correcting data structure for polynomial evaluation.
Our construction follows the non-error-correcting one of~\cite{kedl-umans}
fairly closely; the main new ingredient is to add redundancy to their Chinese Remainder-based reconstruction by using more primes, which gives us the error-correcting features we need. 

\paragraph{\textbf{Time-complexity of decoding and encoding.} }
So far we have used the number of bit-probes as a proxy for the actual time the decoder needs for query-answering.
This is fairly standard, and usually justified by the fact that the actual time complexity of decoding is 
not much worse than its number of bit-probes.  This is also the case for our constructions.
For \noun{membership}, it can be shown that the decoder uses $O(1)$ probes and $\polylog(n)$ time
(as do the RLDCs of~\cite{bghsv04}).
For \noun{polynomial evaluation}, the decoder uses $\polylog(s)\log^{1+o(1)}(n)$ probes and $\polylog(sn)$ time.

The efficiency of \emph{encoding}, i.e., the ``pre-processing'' of the data into the form of a data structure, for both our error-correcting data structures  \noun{membership} and \noun{polynomial evaluation} depends on the efficiency of encoding of the RLDC constructions in~\cite{bghsv04}.  This is not addressed explicitly there, and needs further study.

\section{The \noun{Membership} problem\label{sec:Membership problem}}

In this section we construct a data structure for the membership problem
$\mem_{n,s}$. First we describe some of the building blocks that
we need to prove Theorem~\ref{thm:membership}. Our first basic building
block is the relaxed locally decodable code of Ben-Sasson et al.~\cite{bghsv04}
with nearly linear length. Using our terminology, we can restate their
result as follows:

\begin{theorem}[BGHSV~\cite{bghsv04}]\label{thm:bghsv}For every
$\eps\in(0,1/2)$ and $\eta>0$, there exist an integer $t>0$
and reals $c>0$ and $\tau>0$, such that for every $n$
and every $\delta\leq\tau$, the membership problem $\mem_{n,n}$
has a $(t,\delta,\eps,c\delta)$-data structure for $\mem_{n,n}$
of length $O(n^{1+\eta})$. \end{theorem}

Note that by picking the error-rate $\delta$ a sufficiently small
constant, one can set $\lambda=c\delta$ (the fraction of unrecoverable queries) to be very close to $0$.

The other building block that we need is the following one-probe data
structure of Buhrman et al.~\cite{bmrv:bitvectorsj}.

\begin{theorem}[BMRV~\cite{bmrv:bitvectorsj}]\label{thm:bmrv}
For every $\eps\in(0,1/2)$ and for every positive integers $s\leq n$,
there is an $(1,0,\eps,0)$-data structure for $\mem_{n,s}$ of length
$m=\frac{100}{\eps^{2}}s\log n$ bits. \end{theorem}

\emph{Properties of the BMRV encoding:} The encoding can be represented
as a bipartite graph $\mathcal{G}=(L,R,E)$ with $|L|=n$ left vertices
and $|R|=m$ right vertices, and regular left degree $d=\frac{\log n}{\eps}$.
$\mathcal{G}$ is an expander graph: for each set $S\subseteq L$ with
$|S|\leq2s$, its neighborhood $\Gamma(S)$ satisfies $|\Gamma(S)|\geq\left(1-\frac{\eps}{2}\right)|S|d$.
For each assignment of bits to the left vertices with at most $s$
ones, the encoding specifies an assignment of bits to the right vertices.
In other words, each $x\in\cube^{n}$ of weight $|x|\leq s$ corresponds
to an assignment to the left vertices, and the $m$-bit encoding of
$x$ corresponds to an assignment to the right vertices.

For each $i\in[n]$ we write $\Gamma_{i}:=\Gamma(\{i\})$ to denote
the set of neighbors of $i$. A crucial property of the encoding function
$\E_{bmrv}$ is that for every $x$ of weight $|x|\leq s$, for each
$i\in[n]$, if $y=\E_{bmrv}(x)\in\{0,1\}^{m}$ then $\Pr_{j\in\Gamma_{i}}[x_{i}=y_{j}]\geq1-\eps$.
Hence the decoder for this data structure can just probe a random
index $j\in\Gamma_{i}$ and return the resulting bit $y_{j}$. Note
that this construction is not error-correcting at all, since $|\Gamma_{i}|$
errors in the data structure suffice to erase all information about
the $i$-th bit of the encoded $x$.\qed

\medskip

As we mentioned in the Section~\ref{sub:Our results}, by combining
the BMRV encoding with the data structure for $\mem_{n,n}$ from Theorem~\ref{thm:bghsv},
one easily obtains an $(O(1),\delta,\eps,O(\delta))$-data structure 
for $\mem_{n,s}$ of length $O((s\log n)^{1+\eta})$.
However, we can give an even more efficient, error-correcting data
structure of length $O(s^{1+\eta}\log n)$. Our improvement follows 
an approach taken in de~Wolf~\cite{wolf:ecdata}, which
we now describe. For a vector $x\in\cube^{n}$ with $|x|\leq s$,
consider a BMRV structure encoding $20n$ bits into $m$ bits. Now,
from Section~2.3 in~\cite{wolf:ecdata}, the following ``balls
and bins estimate'' is known:

\begin{prop}[From~\cite{wolf:ecdata}]\label{prop:balls vs bins}
For every positive integers $s\leq n$, the BMRV bipartite graph $\mathcal{G}=([20n],[m],E)$
for $\mem_{20n,s}$ with error parameter $\frac{1}{10}$ has the following
property: there exists a partition of $[m]$ into $b=10\log(20n)$
disjoint sets $B_{1},\ldots,B_{b}$ of $10^{3}s$ vertices each, such
that for each $i\in[n]$, there are at least $\frac{b}{4}$ sets $B_{k}$
satisfying $|\Gamma_{i}\cap B_{k}|=1$. \end{prop}

Proposition~\ref{prop:balls vs bins} suggests the following encoding
and decoding procedures. To encode $x$, we rearrange the $m$ bits
of $\E_{bmrv}(x)$ into $\Theta(\log n)$ disjoint blocks of $\Theta(s)$
bits each, according to the partition guaranteed by Proposition~\ref{prop:balls vs bins}.
Then for each block, encode these bits with the error-correcting
data structure (RLDC) from Theorem~\ref{thm:bghsv}. Given a received word $w,$ to decode $i\in[n]$, pick a block $B_{k}$
at random. With probability at least $\frac{1}{4}$, $\Gamma_{i}\cap B_{k}=\{j\}$
for some $j$. Run the RLDC decoder to decode the $j$-th bit of the
$k$-th block of $w$. Since most blocks don't have much higher error-rate
than the average (which is at most $\delta$), with high probability
we recover $\E_{bmrv}(x)_{j}$, which equals $x_{i}$ with high probability.
Finally, we will argue that most queries do not receive a blank symbol
$\perp$ as an answer, using the expansion property
of the BMRV encoding structure. We now proceed with a formal proof
of Theorem~\ref{thm:membership}.

\begin{proof}[Proof of Theorem~\ref{thm:membership}]We only construct
an error-correcting data structure with error probability $0.49$.
By a standard amplification technique 
we can reduce the error probability to any other positive constant (i.e., repeat the decoder $O(\log(1/\eps))$ times).

By Theorem~\ref{thm:bmrv}, there exists an encoder $\E_{bmrv}$
for an $(1,0,\frac{1}{10},0)$-data structure for the membership problem
$\mem_{20n,s}$ of length $m=10^{4}s\log(20n)$. Let $s'=10^{3}s$.
By Theorem~\ref{thm:bghsv}, for every $\eta>0$, for some $t=O(1)$,
and sufficiently small $\delta$, $\mem_{s',s'}$ has an $(t,10^{5}\delta,\frac{1}{100},O(\delta))$-data structure of length $s''=O(s'^{1+\eta})$. Let $\E_{bghsv}$
and $\D_{bghsv}$ be its encoder and decoder, respectively.

\paragraph{\textbf{Encoding.}}

Let $B_{1},\ldots,B_{b}$ be a partition of $[m]$ as guaranteed by
Proposition~\ref{prop:balls vs bins}. For a string $w\in\cube^{m}$,
we abuse notation and write $w=w_{B_{1}}\cdots w_{B_{b}}$ to denote
the string obtained from $w$ by applying the permutation on $[m]$
according to the partition $B_{1},\ldots,B_{b}$. In other words,
$w_{B_{k}}$ is the concatenation of $w_{i}$ where $i\in B_{k}$.
We now describe the encoding process.

Encoder $\E$: on input $x\in\cube^{n}$, $|x|\leq s$, 
\begin{enumerate}
\item Let $y=\E_{bmrv}\left(x0^{19n}\right)$ and write $y=y_{B_{1}}\ldots y_{B_{b}}$. 
\item Output the concatenation $\E(x)=\E_{bghsv}\left(y_{B_{1}}\right)\ldots\E_{bghsv}\left(y_{B_{b}}\right)$. 
\end{enumerate}
The length of $\E(x)$ is $N=b\cdot O(s'^{1+\eta})=O(s^{1+\eta}\log n)$.

\paragraph{\textbf{Decoding.}}

Given a string $w\in\cube^{N}$, we write $w=w^{(1)}\ldots w^{(b)}$,
where for $k\in[b]$, $w^{(k)}$ denotes the $s''$-bit string $w_{s''\cdot(k-1)+1}\ldots w_{s''\cdot k}$.

Decoder $\D$: on input $i$ and with oracle access to a string $w\in\cube^{N}$, 
\begin{enumerate}
\item Pick a random $k\in[b]$. 
\item If $|\Gamma_{i}\cap B_{k}|\neq1$, then output a random bit. \\
 Else, let $\Gamma_{i}\cap B_{k}=\{j\}$. Run and output the answer
given by the decoder $\D_{bghsv}(j)$, with oracle access to the $s''$-bit
string $w^{(k)}$. 
\end{enumerate}
\textbf{Analysis.} Fix $x\in{D}$ and $w\in\cube^{N}$ such that $\Delta(w,\E(x))\leq\delta N$,
where $\delta$ is less than some small constant $\tau$ to be specified later.
We now verify the four conditions of Definition~\ref{def:data structure}.
For Condition~$1$, note that the number of probes the decoder $\D$
makes is the number of probes the decoder $\D_{bghsv}$ makes, which
is at most $t$, a fixed integer.

We now examine Condition~$2$. Fix $i\in[n]$. By Markov's inequality,
for a random $k\in[b]$, the probability that the relative Hamming
distance between $\E\left(y_{B_{k}}\right)$ and $w^{(k)}$ is greater
than $10^{5}\delta$ is at most $10^{-5}$. If $k$ is chosen such
that the fraction of errors in $w^{(k)}$ is at most $10^{5}\delta$
and $\Gamma_{i}\cap B_{k}=\{j\}$, then with probability at least
$0.99$, $\D_{bghsv}$ outputs $y_{j}$ or $\perp$. Let $\beta\geq\frac{1}{4}$
be the fraction of $k\in[b]$ such that $|\Gamma_{i}\cap B_{k}|=1$.
Then \begin{equation}
\Pr[\D(i)\in\{x_{i},\perp\}]\geq(1-\beta)\frac{1}{2}+\beta\frac{99}{100}-\frac{1}{10^{5}}>0.624.\label{eq:cond2}\end{equation}

To prove Condition~$3$, we need the expansion property of the BMRV
structure, as explained after Theorem~\ref{thm:bmrv}. 
For $k\in[b]$, define $G_{k}\subseteq B_{k}$ so that
$j\in G_{k}$ if $\Pr\left[\D_{bghsv}^{w^{(k)}}(j)=y_{j}\right]\geq0.99$.
In other words, $G_{k}$ consists of indices in block $B_{k}$ that
are answered correctly by $\D_{bghsv}$ with high probability. By Theorem~\ref{thm:bghsv},
if the fraction of errors in $w^{(k)}$ is at most $10^{5}\delta$,
then $|G_{k}|\geq(1-c\delta)|B_{k}|$ for some fixed constant
$c$. Set $A=\cup_{k\in[b]}B_{k}\backslash G_{k}$, Since we showed
above that for a $(1-10^{-5})$-fraction of $k\in[b]$, the fractional
number of errors in $w^{(k)}$ is at most $10^{5}\delta$, we have
$|A|\leq c\delta m+10^{-5}m$.

Recall that the BMRV expander has left degree $d=10\log(20n)$. Take
$\delta$ small enough that $|A|<\frac{1}{40}sd$; this determines
the value of $\tau$ of the theorem. We need to show that for any
such small set $A$, most queries $i\in[n]$ are answered correctly
with probability at least 0.51. It suffices to show that for most
$i$, most of the set $\Gamma_{i}$ falls outside of $A$. To this
end, let $B(A)=\{i\in[n]:|\Gamma_{i}\cap A|\geq\frac{d}{10}\}$. We
show that if $A$ is small then $B(A)$ is small.

\begin{claim} For every $A\subseteq[m]$ with $|A|<\frac{sd}{40}$,
it is the case that $|B(A)|<\frac{s}{2}.$ \end{claim}

\begin{proof} Suppose, by way of contradiction, that $B(A)$ contains
a set $W$ of size $s/2$. $W$ is a set of left vertices in the underlying
expander graph $\mathcal{G}$, and since $|W|<2s$, we must have \[
|\Gamma(W)|\geq\left(1-\frac{1}{20}\right)d|W|.\]
 By construction, each vertex in $W$ has at most $\frac{9}{10}d$
neighbors outside $A$. Thus, we can bound the size of $\Gamma(W)$
from above as follows\begin{eqnarray*}
|\Gamma(W)| & \leq & |A|+\frac{9}{10}d|W|\\
 & < & \frac{1}{40}ds+\frac{9}{10}d|W|\\
 & = & \frac{1}{20}d|W|+\frac{9}{10}d|W|\\
 & = & \left(1-\frac{1}{20}\right)d|W|.\end{eqnarray*}
 This is a contradiction. Hence no such $W$ exists and $|B(A)|<\frac{s}{2}$.
\end{proof}

Define $G=[n]\backslash B(A)$ and notice that $|G|>n-\frac{s}{2}$.
It remains to show that each query $i\in G$ is answered correctly
with probability $>0.51$. To this end, we have \begin{eqnarray*}
\Pr[\D(i)=\perp] & \leq & \Pr[\D\mbox{ probes a block with noise-rate}>10^{5}\delta]+\\
 &  & \Pr[\D\mbox{ probes a }j\in A]+\Pr[\D(i)=\perp:\D\mbox{ probes a }j\not\in A]\\
 & \leq & \frac{1}{10^{5}}+\frac{1}{10}+\frac{1}{100}<0.111.\end{eqnarray*}
 Combining with Eq.~(\ref{eq:cond2}), for all $i\in G$ we have
\[
\Pr[\D(i)=x_{i}]=\Pr[\D(i)\in\{x_{i},\perp\}]-\Pr[\D(i)=\perp]\geq0.51.\]
 Finally, Condition~$4$ follows from the corresponding condition
of the data structure for $\mem_{n,n}$. \end{proof}

\section{The \noun{polynomial evaluation} problem\label{sec:polynomial evaluation}}

In this section we prove Theorem~\ref{thm:poly eval}. Given a polynomial
$g$ of degree $s$ over $\Z_{n}$, our goal is to write down a data
structure of length roughly linear in $s\log n$ so that for each
$a\in\Z_{n}$, $g(a)$ can be computed with approximately $\polylog s\cdot\log n$
bit-probes. Our data structure is built on the work of Kedlaya and
Umans~\cite{kedl-umans}. Since we cannot quite use their construction
as a black-box, we first give a high-level overview of our proof,
motivating each of the proof ingredients that we need.

\paragraph{\textbf{Encoding based on reduced polynomials:}}

The most naive construction, by recording $g(a)$ for each $a\in\Z_{n}$,
has length $n\log n$ and answers an evaluation query with $\log n$
bit-probes. As explained in~\cite{kedl-umans}, one can
reduce the length by using the Chinese Remainder Theorem (CRT): If
$P_{1}$ is a collection of distinct primes, then a nonnegative integer
$m<\prod_{p\in P_{1}}p$ is uniquely specified by (and can be reconstructed
efficiently from) the values $[m]_{p}$ for each $p\in P_{1}$, where
$[m]_{p}$ denotes $m\mod p$. 

Consider the value $g(a)$ over $\Z$, which can be bounded above
by $n^{s+2}$, for $a\in\Z_{n}.$ Let $P_{1}$ consist of the first
$\log(n^{s+2})$ primes. For each $p\in P_{1}$, compute the reduced
polynomial $g_{p}:=g\mod p$ and write down $g_{p}(b)$ for each $b\in\Z_{p}$.
Consider the data structure that simply concatenates the evaluation
table of every reduced polynomial. This data structure has length
$|P_{1}|(\max_{p\in P_{1}}p)^{1+o(1)}$, which is $s^{2+o(1)}\log^{2+o(1)}n$
by the Prime Number Theorem (see Fact~\ref{fact:Prime Number Theorem}
in Appendix~\ref{sec:CRT-code}). Note that $g(a)<\prod_{p\in P_{1}}p$. 
So to compute $[g(a)]_{n}$, it suffices to apply CRT to reconstruct
$g(a)$ over $\Z$ from the values $[g(a)]_{p}=g_{p}([a]_{p})$ for
each $p\in P_{1}$. The number of bit-probes is $|P_{1}|\log(\max_{p\in P_{1}}p)$,
which is $s^{1+o(1)}\log^{1+o(1)}n$. 

\paragraph{\textbf{Error-correction with reduced polynomials:}}

The above CRT-based construction has terrible parameters, but it serves
as an important building block from which we can obtain a data structure
with better parameters. For now, we explain how the above CRT-based
encoding can be made error-correcting. One can protect the bits of
the evaluation tables of each reduced polynomial by an RLDC as provided
by Theorem~\ref{thm:bghsv}. However, the evaluation tables can have
non-binary alphabets, and a bit-flip in just one ``entry'' of
an evaluation table can destroy the decoding process. To remedy this,
one can first encode each entry by a standard error-correcting code
and then encode the concatenation of all the tables by an RLDC. This
is encapsulated in Lemma~\ref{lem:non-binary RLDC}, which can be
viewed as a version of Theorem~\ref{thm:bghsv} over non-binary alphabet.
We prove this in Appendix~\ref{sec:Non-binary-alphabets}. 

\begin{lem}\label{lem:non-binary RLDC} Let $f:{D}\times Q\rightarrow\cube^{\ell}$
be a data structure problem. For every $\eps,\eta,\lambda\in(0,1)$,
there exists $\tau\in(0,1)$ such that for every
$\delta\leq\tau$, $f$ has an $(O(1),\delta,\eps,\lambda)$-data
structure of length $O(\left(\ell|Q|\right)^{1+\eta})$. \end{lem}

To apply Lemma~\ref{lem:non-binary RLDC}, let $D$ be the
set of degree-$s$ polynomials over $\Z_{n}$, $Q$ be the set
of all evaluation points of all the reduced polynomials of $g$ (each specified by a pair $(a,p)$), and
the data structure problem $f$ outputs evaluations of some reduced
polynomial of $g$. 

By itself, Lemma~\ref{lem:non-binary RLDC} cannot guarantee resilience
against noise. In order to apply the CRT to reconstruct $g(a)$, all
the values $\{[g(a)]_{p}:p\in P_{1}\}$ must be correct, which is
not guaranteed by Lemma~\ref{lem:non-binary RLDC}. To fix this,
we add redundancy, taking a larger set of primes than necessary so that the reconstruction
via CRT can be made error-correcting. Specifically, we apply a Chinese
Remainder Code, or CRT code for short, to the encoding process. 

\begin{defn}[CRT code]\label{def:crtcode}Let $p_{1}<p_{2}<\ldots<p_{N}$
be distinct primes, $K<N$, and $T=\prod\limits _{i=1}^{K}p_{i}$.
The \emph{Chinese Remainder Code (CRT code)} with basis $p_{1},\ldots,p_{N}$
and rate $\frac{K}{N}$ over message space $\Z_{T}$ encodes $m\in\Z_{T}$ 
as $\langle[m]_{p_{1}},[m]_{p_{2}},\ldots,[m]_{p_{N}}\rangle$.
\end{defn}

\begin{rem} By CRT, for distinct $m_{1},m_{2}\in\Z_{T}$, their encodings
agree on at most $K-1$ coordinates. Hence the Chinese Remainder Code
with basis $p_{1}<\ldots<p_{N}$ and rate $\frac{K}{N}$ has distance
$N-K+1$. \end{rem}

It is known that good families of CRT code exist and that
unique decoding algorithms for CRT codes (see e.g., \cite{grs:crt-errors})
can correct up to almost half of the distance of the code. The following
statement can be easily derived from known facts, and we include a
proof in Appendix~\ref{sec:CRT-code}. 

\begin{theorem}\label{thm:CRT code}For every positive integer $T$,
there exists a set $P$ consisting of distinct primes, with (1) $|P|=O(\log T),$
and (2) $\forall p\in P,$ $\log T<p<500\log T$, such that a CRT
code with basis $P$ and message space $\Z_{T}$ has rate $\frac{1}{2}$,
and can correct up to a $(\frac{1}{4}-O(\frac{1}{\log\log T}))$-fraction of errors. 
\end{theorem}

We apply Theorem~\ref{thm:CRT code} to a message space of size $n^{s+2}$ to obtain 
a set of primes $P_{1}$ with the properties described above.
Note that these primes are all within a constant
factor of one another, and in particular, the evaluation table of
each reduced polynomial has the same length, up to a constant factor.
This fact and Lemma~\ref{lem:non-binary RLDC} will ensure that our CRT-based
encoding is error-correcting.

\paragraph{\textbf{Reducing the bit-probe complexity:}}

We now explain how to reduce the bit-probe complexity of the CRT-based
encoding, using an idea from~\cite{kedl-umans}. Write $s=d^{m}$, where
$d=\log^{C}s$, $m=\frac{\log s}{C\log\log s}$, and $C>1$ is a sufficiently
large constant. Consider the following multilinear extension map $\psi_{d,m}:\Z_{n}[X]\rightarrow\Z_{n}[X_{0},\ldots,X_{m-1}]$
that sends a univariate polynomial of degree at most $s$ to an $m$-variate
polynomial of degree less than $d$ in each variable. For every $i\in[s]$,
write $i=\sum_{j=0}^{m-1}i_{j}d^{j}$ in base $d$. Define $\psi_{d,m}$
which sends $X^{i}$ to $X_{1}^{i_{0}}\cdots X_{m}^{i_{m-1}}$ and extends
multilinearly to $\Z_{n}[X]$. 

To simplify our notation, we write $\tilde{g}$ to denote the multivariate
polynomial $\psi_{d,m}(g)$. For every $a\in\Z_{n}$, define $\tilde{a}\in\Z_{n}^{m}$
to be $([a]_{n},[a^{d}]_{n},[a^{d^{2}}]_{n},\ldots,[a^{d^{m-1}}]_{n})$.
Note that for every $a\in\Z_{n}$, $g(a)=\tilde{g}(\tilde{a})$ (mod $n$). Now
the trick is to observe that the total degree of the multilinear polynomial
$\tilde{g}$ is less than the degree of the univariate polynomial $g$,
and hence its maximal value over the integers is much reduced.
In particular, for every $a\in\Z_{n}^{m}$, the value $\psi_{d,m}(g)(a)$
over the integers is bounded above by $d^{m}n^{dm+1}$.

We now work with the reduced polynomials of $\tilde{g}$ for our encoding.
Let $P_{1}$ be the collection of primes guaranteed by Theorem~\ref{thm:CRT code}
when $T_{1}=d^{m}n^{dm+1}$. For $p\in P_{1}$, let $\tilde{g}_{p}$
denote $\tilde{g}\mod p$ and $\tilde{a}_{p}$ denote the point $([a]_{p},[a^{d}]_{p},\ldots,[a^{d^{m-1}}]_{p})$.
Consider the data structure that concatenates the evaluation table
of $\tilde{g}_{p}$ for each $p\in P_{1}$. For each $a\in\Z_{n}$,
to compute $g(a)$, it suffices to compute $\tilde{g}(\tilde{a})$
over $\Z$, which by Theorem~\ref{thm:CRT code} can be reconstructed
(even with noise) from the set $\{\tilde{g}_{p}(\tilde{a}_{p}):p\in P_{1}\}$. 

Since the maximum value of $\tilde{g}$ is at most $T_{1}=d^{m}n^{dm+1}$
(whereas the maximum value of $g$ is at most $d^{m}n^{d^{m}+1}$),
the number of primes we now use is significantly less. This effectively
reduces the bit-probe complexity. In particular, each evaluation query
can be answered with $|P_{1}|\cdot\max_{p\in P_{1}}\log p=(dm\log n)^{1+o(1)}$
bit-probes, which by our choice of $d$ and $m$ is equal to $\polylog s\cdot\log^{1+o(1)}n$.
However, the \emph{length} of this encoding is still far from the information-theoretically
optimal $s\log n$ bits. We shall explain how to reduce the length,
but since encoding with multilinear reduced polynomials introduces
potential complications in error-correction, we first explain how
to circumvent these complications.

\paragraph{\textbf{Error-correction with reduced multivariate polynomials:}}

There are two complications that arise from encoding with reduced
multivariate polynomials. The first is that not all the points in
the evaluation tables are used in the reconstructive CRT algorithm.
Lemma~\ref{lem:non-binary RLDC} only guarantees that most of the
entries of the table can be decoded, not all of them. So if the entries
that are used in the reconstruction via CRT are not decoded by Lemma~\ref{lem:non-binary RLDC},
then the whole decoding procedure fails. 

More specifically, to reconstruct $\tilde{g}(\tilde{a})$ over $\Z_{n}$,
it suffices to query the point $\tilde{a}_{p}$ in the evaluation
table of $\tilde{g}_{p}$ for each $p\in P_{1}$. Typically the set $\{\tilde{a}_{p} : a\in\Z_{n}\}$ 
will be much smaller than $\Z_{p}^{m}$, so not all the points in $\Z_{p}^{m}$
are used. To circumvent this issue, we only store the query points
that are used in the CRT reconstruction. Let $B^{p}=\{\tilde{a}_{p}:a\in\Z_{n}\}$.
For each $p\in P_{1}$, the encoding only stores the evaluation of
$\tilde{g}_{p}$ at the points $B^{p}$ instead of the entire domain
$\Z_{p}^{m}$. The disadvantage of computing the evaluation at the
points in $B^{p}$ is that the encoding stage takes time proportional
to $n$. We thus give up on encoding efficiency (which was one of the main goals
of Kedlaya and Umans) in order to guarantee error-correction.

The second complication is that the sizes of the evaluation tables
may no longer be within a constant factor of each other. (This is
true even if the evaluation points come from all of $\Z_{p}^{m}$.)
If one of the tables has length significantly longer than the others,
then a constant fraction of noise may completely corrupt the entries
of all the other small tables, rendering decoding via CRT impossible.
This potential problem is easy to fix; we apply a repetition code
to each evaluation table so that all the tables have equal length.

\paragraph{\textbf{Reducing the length:}}
Now we explain how to reduce the length of the data structure to nearly
$s\log n$, along the lines of Kedlaya and Umans~\cite{kedl-umans}. To reduce the length, we need to reduce the magnitude
of the primes used by the CRT reconstruction. We can effectively achieve
that by applying the CRT twice. Instead of storing the evaluation
table of $\tilde{g}_{p}$, we apply CRT again and store evaluation
tables of the reduced polynomials of $\tilde{g}_{p}$ instead. Whenever
an entry of $\tilde{g}_{p}$ is needed, we can apply the CRT reconstruction to the
reduced polynomials of $\tilde{g}_{p}$. 

Note that for $p_{1}\in P_{1}$, the maximum value of $\tilde{g}_{p_{1}}$ (over the integers rather than mod $n$)
is at most $T_{2}=d^{m}p_{1}^{dm+1}$. Now apply Theorem~\ref{thm:CRT code}
with $T_{2}$ the size of the message space to obtain a collection
of primes $P_{2}$. Recall that each $p_{1}\in P_{1}$ is at most
$O(dm\log n)$. So each $p_{2}\in P_{2}$ is at most $O((dm)^{1+o(1)}\log\log n)$,
which also bounds the cardinality of $P_{2}$ from above. 

For each query, the number of bit-probes made is at most $|P_{1}||P_{2}|\max_{p_{2}\in P_{2}}\log p_{2}$,
which is at most $(dm)^{2+o(1)}\log^{1+o(1)}n$. Recall that by our
choice $d=\log^{C}s$ and $m=\frac{\log s}{C\log\log s}$, we have $dm=\frac{\log^{C+1}s}{C\log\log s}$. Thus,
the bit-probe complexity is $\polylog s\cdot\log^{1+o(1)}n$.

Next we bound the length of the encoding.
Recall that by the remark following Theorem~\ref{thm:poly eval},
we may assume without loss of generality that $s=\Omega(\log^{\zeta}n)$
for some $0<\zeta<1$. This implies $\log\log n=O(\log s)$.
Then for each $p_{2}\in P_{2}$,
$$
p_{2}^{m} \leq \left(O\left((dm)^{1+o(1)}\log\log n\right)\right)^{m}
 \leq (dm)^{(1+o(1))m}\cdot s^{\frac{1}{C}+o(1)} \leq s^{1+\frac{2}{C}+o(1)}.
$$
Now, by Lemma~\ref{lem:non-binary RLDC}, the length of the encoding
is nearly linear in $|P_{1}||P_{2}|\max_{p_{2}\in P_{2}}p_{2}^{m}\log p_{2}$,
which is at most $\polylog s\cdot\log^{1+o(1)}n\cdot\max_{p_{2}\in P_{2}}p_{2}^{m}$.
Putting everything together, the length of the encoding is nearly
linear in $s\log n$. We now proceed with a formal proof.

\begin{proof}[Proof of Theorem~\ref{thm:poly eval}]We only construct
an error-correcting data structure with error probability $\eps=\frac{1}{4}$.
By a standard amplification technique (i.e., $O(\log(1/\eps))$ repetitions)
we can reduce the error probability to any other positive constant.
We now give a formal description of the encoding and decoding algorithms.

\paragraph{\textbf{Encoding:}}
Apply Theorem~\ref{thm:CRT code} with $T=d^{m}n^{dm+1}$ to obtain a collection of primes $P_1$.
Apply Theorem~\ref{thm:CRT code} with $T=d^{m}(\max_{p\in P_{1}}p)^{dm+1}$ to obtain a collection of primes $P_2$. 
Set $p_{max}=\max_{p_{2}\in P_{2}}p_{2}$. 

Now, for each $p_1\in P_1$, $p_2 \in P_2$, define a collection of evaluation points $B^{p_{1},p_{2}}=\{\tilde{a}_{p_{1},p_{2}}:a\in\Z_{n}\}$. 
Fix a univariate polynomial $g\in\Z_{n}[x]$ of degree at most $s$. 
For every $p_{1}\in P_{1}$, $p_{2}\in P_{2}$,
view each evaluation of the reduced multivariate polynomial $\tilde{g}_{p_{1},p_{2}}$
as a bit-string of length exactly $\left\lceil \log p_{max}\right\rceil $.
Let $L=\max_{p_{1}\in P_1, p_{2}\in P_2}|B^{p_{1},p_{2}}|$ and for each $p_1\in P_1$, $p_2 \in P_2$, set $r^{p_{1},p_{2}}=\left\lceil \frac{L}{|B^{p_{1},p_{2}}|}\right\rceil.$
Define $f^{p_{1},p_{2}}$ to be the concatenation of $r^{p_{1},p_{2}}$ copies of 
the string $\langle\tilde{g}(q)\rangle_{q\in B^{p_{1},p_{2}}}$.
Define the string $f=\langle f^{p_{1},p_{2}}\rangle_{p_{1}\in P_1,p_{2}\in P_2}.$ 

We want to apply Lemma~\ref{lem:non-binary RLDC} to protect the
string $f$, which we can since $f$ may be viewed as a data structure
problem, as follows. 
The set of data-items is the set of polynomials $g$ as above.
The set of queries $Q$ is $\bigcup\limits_{p_{1}\in P_{1},p_2\in P_2}B^{p_{1},p_{2}}\times [r^{p_{1},p_{2}}]$.
The answer to query $(q^{p_{1},p_{2}},i^{p_{1},p_{2}})$
is the $i^{p_{1},p_{2}}$-th copy of $\tilde{g}_{p_{1},p_{2}}(q^{p_{1},p_{2}})$.

Fix $\lambda \in (0,1)$. By Lemma~\ref{lem:non-binary RLDC},
for every $\eta>0,$ there exists $\tau_{0}\in(0,1)$ such that for every $\delta\leq\tau_{0},$
the data structure problem corresponding to $f$ has a $(O(\log p_{max}),\delta,2^{-10},\lambda^3 2^{-36})$-data
structure. Let $\E_{0},\D_{0}$ be its encoder and decoder, respectively. 
Finally, the encoding of the polynomial $g$ is simply \[
\E(g)=\E_{0}(f).\]

Note that the length of $\E(g)$ is at most $(|P_{1}||P_{2}|\max_{p_{2}\in P_{2}}p_{2}^{m}\log p_{2})^{1+\eta}$,
which as we computed earlier is bounded above by $O((s \log n)^{1+\zeta})$ for some arbitrarily small constant $\zeta$.

\paragraph{\textbf{Decoding:}}
We may assume, without loss of generality, that the CRT decoder $\D_{crt}$
from Theorem~\ref{lem:non-binary RLDC} outputs $\perp$ when more
than a $\frac{1}{16}$-fraction of its inputs are erasures (i.e., $\perp$ symbols).

The decoder $\D$, with input $a\in\Z_{n}$ and oracle access to $w$,
does the following: 
\begin{enumerate}
\item Compute $\tilde{a}=(a,a^{d},\ldots,a^{d^{m-1}})\in\Z_{n}^{m}$ , and
for every $p_{1}\in P_{1}$, $p_{2}\in P_{2}$, compute the reduced
evaluation points $\tilde{a}_{p_{1},p_{2}}$. 
\item For every $p_{1}\in P_{1}$, $p_{2}\in P_{2}$, pick $j\in[r^{p_{1},p_{2}}]$
uniformly at random and run the decoder $\D_{0}$ with oracle access
to $w$ to obtain the answers $v_{p_{1},p_{2}}^{(a)}=\D_{0}(\tilde{a}_{p_{1},p_{2}},j)$. 
\item For every $p_{1}\in P_{1}$ obtain 
$\displaystyle v_{p_{1}}^{(a)}=\D_{crt}\left(\left(v_{p_{1,}p_{2}}^{(a)}\right)_{p_{2}\in P_{2}}\right).$
\item Output $v^{(a)}=\D_{crt}\left(\left(v_{p_{1}}^{(a)}\right)_{p_{1}\in P_{1}}\right)$. 
\end{enumerate}

\paragraph{\textbf{Analysis:}}
Fix a polynomial $g$ with degree at most $s$. Fix a bit-string $w$
at relative Hamming distance at most $\delta$ from $\E(g)$, where
$\delta$ is at most $\tau_0$.
We proceed to verify that the above encoding and decoding satisfy
the conditions of Definition~\ref{def:data structure}.

Conditions~1 and~4 are easily verified. For Condition~1, observe
that for each $p_{1}\in P_{1}$, $p_{2}\in P_{2}$, $\D_{0}$ makes
at most $O(\log p_{max})$ bit-probes. So $\D$ makes at most $O(|P_{1}||P_{2}|\log p_{max})$
bit-probes, which as we calculated earlier is at most $\polylog s\cdot\log^{1+o(1)}n$. 

For Condition~4, note that since $\D_{0}$ decodes correctly when
no noise is present, $v_{p_{1},p_{2}}^{(a)}$ is equal to $\tilde{g}_{p_{1},p_{2}}(\tilde{a}_{p_{1},p_{2}})$.
By our choice of $P_1$ and $P_2$, after two applications of the Chinese Remainder Theorem, it is easy to see that
$\D$ outputs $v=\tilde{g}(\tilde{a})$, which equals $g(a)$.

Now we verify Condition~2. Fix $a\in\Z_{n}.$ We want to show that
with oracle access to $w$, with probability at least $\frac{3}{4}$, the decoder $\D$ on input $a$ outputs either $g(a)$
or $\perp$.
For $\pi\in P_{1}\cup(P_{1}\times P_{2}),$
we say that a point $v_{\pi}^{(a)}$ is \emph{incorrect} if $v_{\pi}^{(a)}\notin\{\tilde{g}_{\pi}(\tilde{a}_{\pi}),\perp\}$. 

By Lemma~\ref{lem:non-binary RLDC}, for each $p_{1} \in P_1$ and $p_{2}\in P_2$, $v_{p_{1},p_{2}}^{(a)}$ is incorrect with probability at most $2^{-10}$. Now fix $p_{1}\in P_{1}$. On expectation (over the decoder's randomness), at
most a $2^{-10}$-fraction of the points in the set $\{v_{p_{1},p_{2}}^{(a)}:p_{2}\in P_{2}\}$
are incorrect. By Markov's inequality, with probability at least $1-2^{-6}$,
the fraction of points in the set $\{v_{p_{1},p_{2}}^{(a)}:p_{2}\in P_{2}\}$
that are incorrect is at most $\frac{1}{16}$. If the fraction of
blank symbols in the set $\{v_{p_{1},p_{2}}^{(a)}\}_{p_{2}\in P_{2}}$
is at least $\frac{1}{16}$, then $\D_{crt}$ outputs $\perp$, which
is acceptable. Otherwise, the fraction of errors and erasures (i.e., $\perp$ symbols) in
the set $\{v_{p_{1},p_{2}}^{(a)}:p_{2}\in P_{2}\}$ is at most $\frac{1}{8}$.
By Theorem~\ref{thm:CRT code}, the decoder $\D_{crt}$ will output an incorrect $v_{p_{1}}^{(a)}$
with probability at most $2^{-6}$. 
Thus, on expectation, at most a $2^{-6}$-fraction of the points in $\{v_{p_1}^{(a)}:p_1 \in P_1\}$ are incorrect.
By Markov's inequality again, with probability at least $\frac{3}{4}$, at most a $\frac{1}{16}$-fraction of the points in  $\{v_{p_1}^{(a)}:p_1 \in P_1\}$ are incorrect, which by Theorem~\ref{thm:CRT code} implies that $\D^w_a$ is either $\perp$ or $g(a)$.
This establishes Condition~2.

We now proceed to prove Condition~3. We show the existence of a set
$G\subseteq\Z_{n}$ such that $|G|\geq(1-\lambda) n$ and for each $a\in G$,
we have $\Pr[\D(a)=g(a)]\geq\frac{3}{4}$.
Our proof relies on the following observation: for any $p_{1}\in P_{1}$
and $p_{2}\in P_{2}$, if $a\in\Z_{n}$ is chosen uniformly at random,
then the evaluation point $\tilde{a}_{p_{1},p_{2}}$ is like a uniformly chosen element $q\in B^{p_{1},p_{2}}$. 
This observation implies that if a few entries in the evaluation tables of the
multivariate reduced polynomials are corrupted, then for most $a\in\Z_n$, the output of the decoder $\D$ on input $a$ remains unaffected. We now formalize this observation. 

\begin{claim}\label{claim:crt-random}
Fix $p_{1}\in P_{1}$, $p_{2}\in P_{2}$, and a point
$q\in B^{p_{1},p_{2}}$. Then \[
\Pr_{a\in\Z_{n}}\left[\tilde{a}_{p_{1},p_{2}}\equiv q\right]\leq\frac{4}{p_{2}}.\]
\end{claim}

\begin{proof} For any pair of positive integers $m\leq n$, the number
of integers in $[n]$ congruent to a fixed integer mod $m$ is at
most $\left\lfloor \frac{n}{m}\right\rfloor +1$ and at least $\left\lfloor \frac{n}{m}\right\rfloor -1$. 
Note that if $a,b\in\Z_{n}$ with $a\equiv b\mod m$, then for any integer $i$, $a^{i}\equiv b^{i}\mod m$.
Thus, $\tilde{a}_{m}\equiv\tilde{b}_{m}.$ 

It is not hard to see that for a fixed $q_{1}\in B^{p_{1}}$,
the number of integers $a\in\Z_{n}$ such that ${\tilde{a}}_{p_{1}}\equiv q_{1}$
is at most $\left\lfloor \frac{n}{p_{1}}\right\rfloor +1$. Furthermore,
for a fixed $q_{2}\in B^{p_{1},p_{2}}$, the number of points in $B^{p_{1}}$
that are congruent to $q_{2}$ mod $p_{2}$ is at most $\left\lfloor \frac{p_{1}}{p_{2}}\right\rfloor +1$.
Thus, for a fixed $q\in B^{p_{1},p_{2}}$, the number of integers
$a\in\Z_{n}$ such that ${\tilde{a}}_{p_{1},p_{2}}\equiv q$ is at
most $\left(\left\lfloor \frac{n}{p_{1}}\right\rfloor +1\right)\left(\left\lfloor \frac{p_{1}}{p_{2}}\right\rfloor +1\right)$,
which is at most $4\frac{n}{p_{2}}$ since $n\geq p_{1}\geq p_{2}$.
\end{proof}

Now, for every $p_{1}\in P_{1}$ and $p_{2}\in P_{2}$, we say that
a query $(q,j)\in B^{p_{1},p_{2}}\times[r^{p_{1},p_{2}}]$ is \emph{bad}
if the probability that $\D_{0}^{w}(q,j)\neq{\tilde{g}}_{(p_{1},p_{2})}(q)$
is greater than $2^{-10}$. By Lemma~\ref{lem:non-binary RLDC},
the fraction of bad queries in $\cup_{p_{1},p_{2}}B^{p_{1},p_{2}}\times[r^{p_{1},p_{2}}]$
is at most $\lambda_0:=\lambda^3 2^{-36}$. We say that a tuple of primes
$(p_{1},p_{2})\in P_{1}\times P_{2}$ is \emph{bad} if more than a $2^{11}\lambda_0\lambda^{-1}$-fraction
of queries in $B^{p_{1},p_{2}}\times[r^{p_{1},p_{2}}]$ are bad (below,
\emph{good} always denotes not bad.) By averaging, the fraction of
bad tuples $(p_{1},p_{2})$ is at most $2^{-11}\lambda$.

For a fixed good tuple $(p_{1},p_{2})$, we say that an index $i^{p_{1},p_{2}}$
is \emph{bad} if more than a $2^{-11}\lambda$-fraction of queries in the
copy $B^{p_{1},p_{2}}\times\{i^{p_{1},p_{2}}\}$ are bad. Since $(p_{1},p_{2})$
is good, by averaging, at most a $2^{22}\lambda_0\lambda^{-2}$-fraction
of $[r^{p_{1},p_{2}}]$ are bad. Recall that in Step~2 of the decoder
$\D$, the indices $\{j^{p_{1},p_{2}}:p_{1}\in P_{1},p_{2}\in P_{2}\}$
are chosen uniformly at random. So on expectation, the set of indices
$\{j^{p_{1},p_{2}}:(p_{1},p_{2})\mbox{ is good}\}$ has at most a $2^{22}\lambda_0\lambda^{-2}$-fraction
of bad indices. By Markov's inequality, with probability at least
$\frac{7}{8}$, the fraction of bad indices in the set $\{j^{p_{1},p_{2}}:(p_{1},p_{2})\mbox{ is good}\}$
is at most $2^{25}\lambda_0\lambda^{-2}$. We condition on
this event occurring and fix the indices $j^{p_{1},p_{2}}$ for each
$p_{1}\in P_{1}$, $p_{2}\in P_{2}$.

Fix a good tuple $(p_{1},p_{2})$ and a good index $j^{p_{1},p_{2}}$.
By Claim~\ref{claim:crt-random}, for a uniformly random $a\in\Z_{n}$,
the query $(\tilde{a}_{p_{1},p_{2}},j^{p_{1},p_{2}})$ is bad with
probability at most $2^{-9} \lambda$. By linearity of expectation, for
a random $a\in\Z_{n}$, the expected fraction of bad queries in the
set $S^{a}=\{(\tilde{a}_{p_{1},p_{2}},j^{p_{1},p_{2}}):p_1\in P_1,p_2\in P_2\}$
is at most $2^{-11}\lambda+2^{25}\lambda_0\lambda^{-2}+2^{-9}\lambda$, which
is at most $2^{-8}\lambda$ by definition of $\lambda_0$. Thus, by Markov's
inequality, for a random $a\in\Z_{n}$, with probability at least
$1-\lambda$, the fraction of bad queries in the set $S^{a}$ is at most
$2^{-8}$. By linearity of expectation, there exists some subset $G\subseteq\Z_{n}$
with $|G|\geq(1-\lambda)n$ such that for every $a\in G$, the fraction of
bad queries in $S^{a}$ is at most $2^{-8}$. 

Now fix $a\in G$. By definition, the fraction of bad queries in $S^{a}$ is at
most $2^{-8}$, and furthermore, each of the good queries in $S^{a}$ is incorrect with probability
at most $2^{-10}$. So on expectation, the fraction of errors and
erasures in $S^{a}$ is at most $2^{-8}+2^{-10}$. By Markov's inequality, with
probability at least $\frac{7}{8}$, the fraction of errors and
erasures in the set $\{v_{p_{1},p_{2}}^{(a)}:p_{1}\in P_{1},p_{2}\in P_{2}\}$
is at most $2^{-5}+2^{-7}$, which is at most $\frac{1}{25}$. 
We condition on this event occurring. By averaging,
for more than a $\frac{4}{5}$-fraction of the primes $p_{1}\in P_{1}$,
the set $\{v_{p_{1},p_{2}}^{(a)}:p_{2}\in P_{2}\}$ has at most 
$\frac{1}{5}$-fraction of errors and erasures, which can be corrected
by the CRT decoder $\D_{crt}$. Thus, after Step~3 of the decoder
$\D$, the set $\{v_{p_{1}}^{(a)}\}$ has at most a $\frac{1}{5}$-fraction
of errors and erasures, which again will be corrected by the CRT decoder
$\D_{crt}$. Hence, by the union bound, the two events that we conditioned on
earlier occur simultaneously with probability at least $\frac{3}{4}$, and $\D(a)$ will
output $g(a)$. 
\end{proof}

\section{Conclusion and future work\label{sec:Conclusion}}
We presented a relaxation of the notion of error-correcting data structures
recently proposed in~\cite{wolf:ecdata}. While the earlier definition
does not allow data structures that are both error-correcting and
efficient in time and space (unless an unexpected breakthrough happens for
constant-probe LDCs), our new definition allows us to construct
efficient, error-correcting data structures for both the \noun{membership}
and the \noun{polynomial evaluation} problems. This opens up many
directions: what other data structures can be made error-correcting?

The problem of computing \emph{rank} within a sparse ordered set is a good target.
Suppose we are given a universe $[n]$, some nonnegative
integer $s\leq n$, and a subset $S\subseteq[n]$ of size at most
$s$. The rank problem is to store $S$ compactly so that on input
$i\in[n]$, the value $|\{j\in S:j\leq i\}|$ can be computed efficiently.
For easy information-theoretic reasons, any data structure for this problem 
needs length at least $\Omega(s\log n)$ and makes
$\Omega(\log s)$ bit-probes for each query. If $s=O(\log n)$, one
can trivially obtain an error-correcting data structure of optimal length $O(s\log n)$
with $O(\log^{2}n)$ bit-probes, which is only quadratically worse than optimal: 
write down $S$ as a string of $s\log n$ bits, encode it with a good error-correcting code, 
and read the entire encoding when an index is queried.
However, it may be possible to do something smarter and more involved.
We leave the construction of near-optimal error-correcting data structures for rank with small $s$
(as well as for related problems such as \emph{predecessor}) as challenging open problems.

\subsection*{Acknowledgments}
We thank Madhu Sudan for helpful comments and suggestions on the presentation of this paper.

\bibliographystyle{plain}
\bibliography{qc}

\appendix
\section{Non-binary answer set\label{sec:Non-binary-alphabets}}
We prove Lemma~\ref{lem:non-binary RLDC}, a version of Theorem~\ref{thm:bghsv}
when the answer set $A$ is non-binary. We first encode the $\ell|Q|$-bit
string $\left\langle f(x,q)\right\rangle_{q\in Q}$ by an RLDC, and
use the decoder of the RLDC to recover each of the $\ell$ bits of
$f(x,q)$. Now it is possible that for each $q\in Q$, the decoder
outputs some blank symbols $\perp$ for some of the bits of $f(x,q)$,
and no query could be answered correctly. To circumvent this, we first
encode each $\ell$-bit string $f(x,q)$ with a good error-correcting
code, then encode the entire string by the RLDC. Now if the decoder
does not output too many errors or blank symbols among the bits of
the error-correcting code for $f(x,q)$, we can recover it. We need
a family of error-correcting codes with the following property, see
e.g. page $668$ in~\cite{huffman-coding}.

\begin{fact}\label{fact:good code} For every $\delta\in(0,1/2)$
there exists $R\in(0,1)$ such that for all $n$, there exists a binary
linear code of block length $n$, information length $Rn$, Hamming
distance $\delta n$, such that the code can correct from $e$ errors
and $s$ erasures, as long as $2e+s<\delta n$. \end{fact}

\begin{proof}[Proof of Lemma~\ref{lem:non-binary RLDC}]
We only construct an error-correcting data structure with error probability $\eps=\frac{1}{4}$.
By a standard amplification technique (i.e., $O(\log(1/\eps))$ repetitions)
we can reduce the error probability to any other positive constant.
Let $\E_{ecc}:\cube^{\ell}\rightarrow\cube^{\ell'}$ be an asymptotically
good binary error-correcting code (from Fact~\ref{fact:good code}),
with $\ell'=O(\ell)$ and relative distance $\frac{3}{8}$, and decoder
$\D_{ecc}$. By Theorem~\ref{thm:bghsv}, there exist $c_{0},\tau_{0}>0$
such that for every $\delta\leq\tau_{0}$, there is a $(O(1),\delta,\frac{1}{32},c_{0}\delta)$-relaxed locally decodable code (RLDC).
Let $\E_{0}$ and $\D_{0}$ denote its encoder and decoder, respectively.

\paragraph{\textbf{Encoding.}}

We construct a data structure for $f$ as follows. Define the encoder
$\E:D\rightarrow\cube^{N}$, where $N=O(\left(\ell'\cdot|Q|\right)^{1+\eta})$,
as \[
\E(x)=\E_{0}\left(\left\langle \,\E_{ecc}(f(x,q))\,\right\rangle _{q\in Q}\right).\]

\paragraph{\textbf{Decoding.}}

Without loss of generality, we may impose an ordering on the set $Q$
and identify each $q\in Q$ with an integer in $[Q]$.

The decoder $\D$, with input $q\in Q$ and oracle access to $w\in\cube^{N}$,
does the following: 
\begin{enumerate}
\item For each $j\in[\ell']$, let $r_{j}=\D_{0}^{w}\left((q-1)\ell'+j\right)$
and set $r=r_{1}\ldots r_{\ell'}\in\{0,1,\perp\}^{\ell'}$. 
\item If the number of blank symbols $\perp$ in $r$ is at least $\frac{\ell'}{8}$,
then output $\perp$. Else, output $\D_{ecc}(r)$. 
\end{enumerate}

\paragraph{\textbf{Analysis.}}

Fix $x\in{D}$ and $w\in\cube^{N}$ such that $\Delta(w,\E(x))\leq\delta N$,
and $\delta\leq\tau$, where $\tau$ is the minimum of $\tau_{0}$
and $\lambda2^{-6}c_0^{-1}$. We need to argue that the above encoding
and decoding satisfies the four conditions of Definition~\ref{def:data structure}.
For Condition~$1$, since $\D_{0}$ makes $O(1)$ bit-probes and $\D$
runs this $\ell'$ times, $\D$ makes $O(\ell')=O(\ell)$ bit-probes into
$w$.

We now show $\D$ satisfies Condition~$2$. Fix $q\in Q$. We want
to show $\Pr[\D^{w}(q)\in\{f(x,q),\perp\}]\geq\frac{3}{4}$. By Theorem~\ref{thm:bghsv},
for each $j\in[\ell']$, with probability at most $\frac{1}{32}$,
$r_{j}=f(x,q)_{j}\oplus1$. So on expectation, for at most a $\frac{1}{32}$-fraction
of the indices $j$, $r_{j}=f(x,q)_{j}\oplus1$. By Markov's inequality,
with probability at least $\frac{3}{4}$, the number of indices $j$ such
that $r_{j}=f(x,q)_{j}\oplus1$ is at most $\frac{\ell'}{8}$. If
the number of $\perp$ symbols in $r$ is at least $\frac{\ell'}{8}$
then $\D$ outputs $\perp$, so assume the number of $\perp$ symbols
is less than $\frac{\ell'}{8}$. Those $\perp$'s are viewed as erasures
in the codeword $\E_{ecc}(f(x,q))$. Since $\E_{ecc}$ has relative
distance $\frac{3}{8}$, by Fact~\ref{fact:good code}, $\D_{ecc}$
will correct these errors and erasures and output $f(x,q)$.

For Condition~$3$, we show there exists a large subset $G$ of $q$'s
satisfying $\Pr[\D^{w}(q)=f(x,q)]\geq\frac{3}{4}$. Let $y=\left\langle \,\E_{ecc}(f(x,q))\,\right\rangle _{q\in Q}$,
which is a $\ell'|Q|$-bit string. Call an index $i$ in $y$ \emph{bad}
if $\Pr[\D_{0}^{w}(i)=y_{i}]<\frac{3}{4}.$ By Theorem~\ref{thm:bghsv},
at most a $c_{0}\delta$-fraction of the indices in $y$ are bad.
We say that a query $q\in Q$ is \emph{bad} if more than a $\frac{1}{64}$-fraction
of the bits in $\E_{ecc}(f(x,q))$ are bad. By averaging, the fraction of bad queries in $Q$ is at most 
$64c_{0}\delta$, which is at most $\lambda$ by our choice of $\tau$. 
We define $G$ to be the set of $q\in Q$ that are not bad. Clearly $|G|\geq (1-\lambda) |Q|$.

Fix $q\in G$. On expectation (over the decoder's randomness), the fraction of indices
in $r$ such that $r_{j}\neq f(x,q)_{j}$ is at most $\frac{1}{64}+\frac{1}{32}$. 
Hence by Markov's inequality,
with probability at least $\frac{3}{4}$, the fraction
of indices in $r$ such that $r_{j}\neq f(x,q)_{j}$ is at most $\frac{3}{16}$. 
Thus, by Fact~\ref{fact:good code}, $\D_{ecc}(r$) will recover from these
errors and erasures and output $f(x,q)$.

Finally, Condition~$4$ follows since the pair $(\E_{0},\D_{0})$
satisfies Condition~4, finishing the proof.
\end{proof} 

\section{CRT codes\label{sec:CRT-code}}

In this section we explain how Theorem~\ref{thm:CRT code} follows from known facts.
In~\cite{grs:crt-errors}, Goldreich, Ron, and Sudan
designed a unique decoding algorithm for CRT code. 

\begin{theorem}[from~\cite{grs:crt-errors}]\label{thm:CRT decoding}
Given a CRT Code with basis $p_{1}<\ldots<p_{N}$ and rate $K/N$,
there exists a polynomial-time algorithm that can correct up to $\frac{\log p_{1}}{\log p_{1}+\log p_{N}}(N-K)$
errors. 
\end{theorem}

By choosing the primes appropriately, we can establish Theorem~\ref{thm:CRT code}.
In particular, the following well-known estimate, essentially a consequence
of the Prime Number Theorem, is useful. See for instance Theorem~4.7
in~\cite{apostol} for more details.

\begin{fact}\label{fact:Prime Number Theorem} For an integer
$\ell>0$, the $\ell$th prime (denoted $q_{\ell}$) satisfies $\frac{1}{6}\ell\log\ell<q_{\ell}<13\ell\log\ell$.
\end{fact}

\begin{proof}[Proof of Theorem \ref{thm:CRT code}] Let $K=\lfloor\frac{12\log T}{\log\log T}\rfloor$
and $q_{\ell}$ denote the $\ell$-th prime. By Fact~\ref{fact:Prime Number Theorem},
$q_{K}>\frac{1}{6}K\log K>\log T$ and $q_{3K-1}<39K\log3K<500\log T$.
Also, notice that $\prod_{i=K}^{2K-1}q_{i}>q_{K}^{K}>(\log T)^{\frac{\log T}{\log\log T}}=T.$
Thus, by Definition \ref{def:crtcode}, the CRT code with basis $q_{K}< \ldots <q_{2K-1}< \ldots < q_{3K-1}$ and message space $\Z_{T}$, has rate at most $\frac{K}{2K}=\frac{1}{2}$.
Lastly, by Theorem~\ref{thm:CRT decoding}, the code can correct a fraction
$\frac{1}{4}-O(\frac{1}{\log\log T})$ of errors.
\end{proof}

\end{document}